\documentclass[twoside]{article}
\usepackage{times}
\usepackage{graphicx}
\usepackage{amsthm,amsmath,amssymb}

\pdfpageheight\paperheight
\pdfpagewidth\paperwidth

\newtheorem{lemma}{Lemma}
\newtheorem{theorem}{Theorem}
\newtheorem{corollary}{Corollary}
\newtheorem{definition}{Definition}

\newcommand\ket[1]{\ensuremath{|#1\rangle}}
\newcommand\bra[1]{\ensuremath{\langle#1|}}

\newcommand\oprod[2]{\ensuremath{|#1\rangle\langle#2|}}
\newcommand\tr{\mathop{\rm tr}\nolimits}

\newcommand\class[1]{{\bf\textup{#1}}}

\title{\Large\bf Non-Identity Check Remains QMA-Complete for Short Circuits}

\author{
  Zhengfeng Ji \\
  \it \small Perimeter Institute for Theoretical Physics, Waterloo Ontario, Canada.
  \and
  Xiaodi Wu \\
  \it \small Institute for Quantum Computing, University of
  Waterloo, Ontario, Canada, \\
  \it \small Department of Electrical Engineering and Computer Science, University of Michigan, Ann Arbor, USA.
}

\date{\normalsize June 30, 2008}

\begin{document}

%% End-Of-Header

\maketitle

\begin{abstract}
  The Non-Identity Check problem asks whether a given a quantum
  circuit is far away from the identity or not. It is well known that
  this problem is \class{QMA}-Complete~\cite{JWB05}. In this note, it
  is shown that the Non-Identity Check problem remains
  \class{QMA}-Complete for circuits of short depth. Specifically, we
  prove that for constant depth quantum circuit in which each gate is
  given to at least $\Omega(\log n)$ bits of precision, the
  Non-Identity Check problem is \class{QMA}-Complete. It also follows
  that the hardness of the problem remains for polylogarithmic depth
  circuit consisting of only gates from any universal gate set and for
  logarithmic depth circuit using some specific universal gate set.
\end{abstract}

\section{Introduction}

\label{sec:intro}

Quantum circuit is the natural quantum analog of classical circuit and
an important model~\cite{Yao93} to analyze the power of quantum
computation. A \emph{quantum circuit} is an acyclic network of
\emph{quantum gates} connected by wires. The quantum gates represent
feasible quantum operations (unitary operations in our model),
involving constant numbers of qubits. The \emph{depth} of a circuit is
the maximum number of quantum gates affecting on any qubit from input
to output.

Much of the difficulty in implementing quantum computation is the
decoherence effect of the qubits which happens in a very short time.
Short depth quantum circuit seems to provide a way to implement as
much quantum computation as possible in very limited available time
due to the decoherence effect. Thus, analyzing the power of short
depth quantum circuit is of significant interest.

A few examples about the power of logarithmic depth quantum circuit
have been proposed in the past few years~\cite{CW00,MN02}. Besides, a
systematic procedure has also been discovered~\cite{BE09} to
parallelize a class of quantum circuits to logarithmic depth. The
investigation of the power of constant depth quantum circuit has also
been started recently~\cite{FGHZ05,TD04}. In this paper, we prove the
hardness of the Non-Identity Check problem for such short depth
quantum circuits.

The \emph{Non-Identity Check} problem is to decide if a quantum
circuit is far away from the identity, given a classical description
of the circuit. More generally, one can ask whether two quantum
circuits $U$ and $V$ are equivalent or not. But is it easy to see that
the equivalence problem can be reduced to the identity check problem
of $UV^\dagger$. Classically, similar
problems~\cite{BCW80,Sch80,Zip79} determine whether two given
classical circuits are equivalent or not. It turns out that the
classical problem can be solved efficiently using a randomized
algorithm. That is, the classical problem is in \class{BPP}. In
contrast, we know that the quantum Non-Identity Check problem is
\class{QMA}-Complete~\cite{JWB05}. This means that the problem is hard
even for quantum computers. Moreover, as will be shown in this paper,
the hardness remains even when only short depth circuits are
considered.

The complexity class \class{QMA} is the quantum version of \class{NP}.
It differs from \class{NP} in that the witness can be a quantum state
and that the verifier has the power of performing polynomial time
quantum computation. A lot has been known about this complexity class.
One of the most important facts is that it has a complete problem
which naturally generalizes the Boolean Satisfiability problem. The
first proof of it by Kitaev~\cite{KSV02} serves as the quantum analog
of the Cook-Levin theorem~\cite{Cook71,Lev73}. The survey~\cite{AN02}
may also be helpful in understanding the original proof.

The \emph{Local Hamiltonian} problem has been the first known
important complete problem for \class{QMA} and has also turned out to
be the most studied one. In fact, the last few years have witnessed a
series of improvements on it~\cite{JWB05,KR03,KKR04,OT05,AGIK09},
culminating in the result that the problem remains complete even for
$1$-D local Hamiltonian. Another complete problem for \class{QMA} is
Non-Identity Check~\cite{JWB05}, which is also the main topic of this
paper. There haven't been many \class{QMA}-Complete problems found. In
addition to the Local Hamiltonian and Non-Identity Check problem, we
also know that the Local Consistency problem and related
variants~\cite{Liu06,LMF07,WMN09} and the Quantum Clique
Problem~\cite{BS07} are \class{QMA}-Complete.

The main result of this paper is that Non-Identity Check for short
quantum circuits remains \class{QMA}-Complete. Formally, we have:

\begin{theorem}\label{thm:const-depth}
  Non-Identity Check of constant depth quantum circuit on $n$ qubits
  is \class{QMA}-Complete if the encoding of the circuit describes
  each gate to at least $\Omega(\log n)$ bit of precision.
\end{theorem}

When a circuit is restricted to consisting of only gates from a finite
universal gate set, we can have the following Corollary, which is a
direct application of the \emph{Solovay-Kitaev Theorem}~\cite{DN05}.

\begin{corollary}\label{cor:polylog-depth}
  Non-Identity Check is \class{QMA}-Complete for
  $O(\log^{\delta}(n))$-depth quantum circuits of an arbitrary
  universal gate set on $n$ qubits where $\delta\approx 3$.
\end{corollary}

Interestingly, there are more efficient universal gate sets as shown
in Ref.~\cite{HRC02}. With these special universal gate sets, we could
have even shorter depth quantum circuits. Precisely,

\begin{corollary}\label{cor:log-depth}
  There exists a universal gate set such that Non-Identity Check is
  \class{QMA}-Complete for logarithmic depth quantum circuits using
  this particular universal gate set.
\end{corollary}

In previous works where the depth is not an issue, it is not necessary
to distinguish whether the encoding of the circuit uses a fixed
universal gate set or not. But this subtlety is the key point that
makes the difference in Theorem~\ref{thm:const-depth} and
Corollary~\ref{cor:polylog-depth}~and~\ref{cor:log-depth}.

To prove Theorem 1, we will employ the $1$-D local Hamilton problem
(\class{QMA}-Complete) as our starting point, and reduce it to a short
circuit Non-Identity Check problem. The reminder of the paper is
organized as follows. In the next section, some definitions and
notations are summarized. In Section~\ref{sec:main}, our main result
is proved. We conclude with Section~\ref{sec:con}.

\section{Preliminary}

\label{sec:pre}

In this section, we explain the notions used in the rest of the paper.

The \emph{spectral norm} $\|A\|$ of matrix $A$ is defined as
\begin{equation*}
  \|A\| = \max_{\ket{\psi}} \frac{\|A\ket{\psi}\|}{\|\ket{\psi}\|},
\end{equation*}
and the \emph{trace norm} $\|A\|_{\tr}$ defined as
\begin{equation*}
  \|A\|_{\tr} = \tr \sqrt{A^\dagger A}.
\end{equation*}
The \emph{numerical range} of a matrix $A$ is the subset of the
complex plain $\{ \bra{\psi} A \ket{\psi}\}$ and is known to be a
convex set. In particular, for normal matrices the numerical range is
simply the convex hull of all eigenvalues. For any Hermitian matrix
$H$, $\lambda_{\max}(H)$ and $\lambda_{\min}(H)$ are the largest and
smallest eigenvalue of $H$. Denote the eigenvalue range of $H$ by
$\lambda(H) = \lambda_{\max}(H)-\lambda_{\min}(H)$.

The eigenvalues of a unitary matrix $U$ lie on the unit circle of the
complex plain. The distribution of the eigenvalues is important to
characterize the closeness of $U$ to identity $I$. See for example the
illustration made in Figure~\ref{fig:dist} where the eigenvalues of
$U$ are marked on the unit circle as small hollow circles. Use
$\alpha_{\max}(U)$ and $\alpha_{\min}(U)$ to denote the maximal and
minimal value of the arguments of eigenvalues of $U$ taken in the
interval $(-\pi, \pi]$. They correspond to the argument of point $A$
and $C$ in Figure~\ref{fig:dist}. Let $\tilde{\alpha}(U)$ be the
length of the shortest arc that contains all eigenvalues of $U$ (which
corresponds to arc $\stackrel{\frown}{AC}$ in the figure). It was
known that $U$ is perfectly distinguishable from $I$ if and only if
$\tilde{\alpha}(U) \ge \pi$ ~\cite{DFY07}. Define a new quantity
called \emph{phase range} as
\begin{equation}
  \label{eq:alpha}
  \alpha(U) = \min \{ \pi, \tilde{\alpha}(U) \},
\end{equation}
and extend it to be defined on two unitary operations $U$ and $V$ as
\begin{equation}
  \label{eq:alpha2}
  \alpha(U,V) = \alpha(U^\dagger V).
\end{equation}

\begin{figure}[!hbt]
  \centering
  \includegraphics{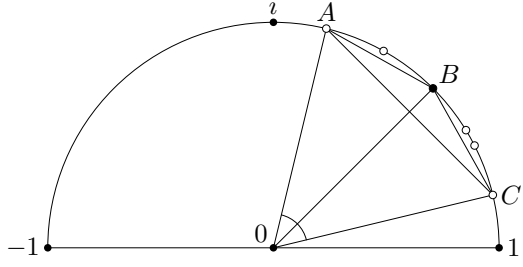}
  \caption{Comparison of different distances}
  \label{fig:dist}
\end{figure}

The diamond norm~\cite{Kit97} serves as a good way of measuring
distance of quantum operations. For a superoperator $\Phi$ mapping
operators acting on Hilbert space $\mathcal{H}_1$ to operators acting
on Hilbert space $\mathcal{H}_2$, define the diamond norm of $\Phi$ as
\begin{equation}
  \|\Phi\|_\diamond = \max_{\rho} \| \Phi\otimes I_{\mathcal{H}_1}
  (\rho) \|_{\tr},
\end{equation}
where the maximum is take over density matrices $\rho$.

Let $\mathcal{U}$ be the quantum operation corresponding to unitary
$U$ as
\begin{equation*}
  \mathcal{U}(\rho) = U\rho U^\dagger.
\end{equation*}
It was known that~\cite{Wat08}
\begin{equation*}
  \|\mathcal{U}-\mathcal{I}\|_\diamond = 2\sqrt{1-\nu^2(U)},
\end{equation*}
where $\mathcal{I}$ is the identity operation and $\nu(U)$ is the
minimum distance of the zero point to the numerical range of $U$. As
$U$ is normal, its numerical range is the convex hull of all of its
eigenvalues and the diamond norm is exactly the length of segment $AC$
in Figure~\ref{fig:dist}. Therefore, we have
\begin{equation*}
  \|\mathcal{U}-\mathcal{I}\|_\diamond = 2 \sin\frac{\alpha(U)}{2}.
\end{equation*}

Another way to measure the closeness of $U$ and $I$ is the following
quantity~\cite{JWB05}:
\begin{equation}\label{eq:norm_dist}
  \min_\varphi \|U-e^{i\varphi} I\|.
\end{equation}
We can also visualize the idea of the definition in
Figure~\ref{fig:dist}. The minimum in Eq.~\eqref{eq:norm_dist} will be
achieved when $\varphi$ is the argument of point $B$ in the middle of
the arc connection $A$ and $C$, and the minimum value is the length of
segment $AB$. Its relation with phase range $\alpha$ when $\alpha(U)<
\pi$ is
\begin{equation*}
  \min_\varphi \|U-e^{i\varphi} I\| = 2\sin\frac{\alpha(U)}{4}.
\end{equation*}
When $\alpha(U) = \pi$, they are not related but we will always have
\begin{equation}\label{eq:sin_alpha}
  \min_\varphi \|U-e^{i\varphi} I\| \ge 2\sin\frac{\alpha(U)}{4}.
\end{equation}

In the rest of this section, we give the definition of complexity
class \class{QMA} and some of its complete problems.

Let $\Sigma$ be the alphabet $\{0, 1\}$ and denote by $|x|$ the length
of string $x$. A family of unitary quantum circuits $\{U_x,
x\in\Sigma^*\}$ is said to be generated in polynomial-time if there is
a classical deterministic Turing machine which, on input $x$, outputs
the encoding of circuit $U_x$ in time polynomial in $|x|$. A circuit
accepts its input state if the first output qubit is measured to be
``1''.

The complexity class \class{QMA} can be defined as follows.

\begin{definition}[QMA]
  A language $L$ is in \class{QMA} if there is a family of circuits
  $\{U_x, x\in\Sigma^*\}$ generated in polynomial-time together with a
  polynomial $m$ such that $U_x$ acts on $m+k$ qubits and the
  following holds:
  \begin{enumerate}
  \item If $x\in L$, there exists an $m(|x|)$-qubit state $\ket{\psi}$
    such that $\Pr \left[U_x \text{ accepts } \ket{\psi} \otimes
      \ket{0}^{\otimes{k(|x|)}} \right] \ge 2/3$;
  \item If $x\not\in L$, for all $m(|x|)$-qubit state $\ket{\psi}$,
    $\Pr \left[U_x \text{ accepts } \ket{\psi} \otimes
      \ket{0}^{\otimes{k(|x|)}} \right] \le 1/3$.
  \end{enumerate}
\end{definition}

\class{QMA} has complete problems. We will make use of the
completeness of the Local Hamiltonian problem, especially its $1$-D
version. Therefore, it will be discussed in more detail although the
main focus of this paper is the Non-Identity Check problem.

Consider a Hamiltonian $H$ of an $n$-particle system with constant
local dimension. $H$ is called $k$-local if it is the sum $\sum_i H_i$
where each $H_i$ acts non-trivially only on $k$ particles. Sometimes,
there is also an underlying layout of the particles in the problem,
for example $1$-D chain or $2$-D lattice, such that each local term
$H_i$ acts only on neighbouring particles corresponding to the layout.
We will call them $1$-D or $2$-D Local Hamiltonian problem
respectively. For $1$-D Hamiltonian $H=\sum H_i$, the particles are
arranged on a line, and each local term $H_i$ acts non-trivially only
on two neighbouring particles.

The general Local Hamiltonian problem can be formalized as in the
following definition.

\begin{definition}[Local Hamiltonian Problem]
  Given a $k$-local Hamiltonian $H=\sum_{i=1}^r H_i$ of $n$ particles
  and two real numbers $a, b$, where $H_i$ has bounded norm and
  $b-a\ge 1/\text{poly}(n)$, $r$ is polynomial in $n$ and $k$ is
  $O(1)$. It is promised that the lowest eigenvalue of $H$ is either
  smaller than $a$ or larger than $b$. Output ``Yes'' in the first
  case and ``No'' otherwise.
\end{definition}

The problem was first shown to be \class{QMA}-Complete for $5$-local
Hamiltonian~\cite{KSV02,AN02}. Recent developments have improved this
to Hamiltonians with much simpler structures -- the $3$-local,
$2$-local, $2$-D, and even $1$-D cases -- all proved to be complete
for \class{QMA}~\cite{KR03,KKR04,OT05,AGIK09}.

Non-Identity Check problem was first considered in Ref.~\cite{JWB05}.
It can be stated as:

\begin{definition}[Non-Identity Check]
  Given a classical description of a quantum circuit $U$ on $n$ qubits
  and two real numbers $a, b$ with $b-a\ge 1/\text{poly}(n)$. It is
  promised that
  \begin{equation*}
    \min\limits_\varphi\; \|U-e^{i\varphi} I\|
  \end{equation*}
  is either larger than $b$ or smaller than $a$. Output ``Yes'' in the
  first case and ``No'' in the second.
\end{definition}

In the definition of the problem, the quantity $\min\limits_\varphi \|
U-e^{i\varphi} I \|$ is used to evaluate the closeness of $U$ to
identity. We can also use phase range $\alpha(U)$ or diamond norm
instead. And it's easy to see that, all the three definitions
mentioned above can be used in defining the Non-Identity Check problem
without changing anything. The point is that they are quantities
related to each other by monotonic trigonometric functions. Moreover,
the inverse polynomial gap in one of them implies that in the others.
In the next section, we will use phase range to define and analyze the
Non-Identity Check problem. It is interesting to note at this point
that the hardness of Non-Identity Check implies that of the estimation
of the diamond norm of the difference of two unitary quantum circuits
to inverse polynomial precision.

\section{Hardness of Non-Identity Check for Short Circuits}

\label{sec:main}

We will prove the hardness of Non-Identity Check problem for short
circuits by reducing the $1$-D Local Hamiltonian problem to it. The
main technical tool is $\alpha,\alpha_{\max},\alpha_{\min}$ discussed
in Section~\ref{sec:pre}. Namely, the following lemmas will be useful
in the proof. The first two can be found in the Appendix of
Ref.~\cite{CPR00} and we won't prove them here.

\begin{lemma}\label{lem:cpr}
  For unitary $U_1$ and $U_2$ such that
  \begin{equation*}
    \begin{split}
      \alpha_{\max} (U_1) + \alpha_{\max} (U_2) & < \pi,\\
      \alpha_{\min} (U_1) + \alpha_{\min} (U_2) & >-\pi,
    \end{split}
  \end{equation*}
  we have
  \begin{equation*}
    \begin{split}
      \alpha_{\max}(U_1U_2) & \le \alpha_{\max}(U_1) +
      \alpha_{\max}(U_2),\\
      \alpha_{\min}(U_1U_2) & \ge \alpha_{\min}(U_1) +
      \alpha_{\min}(U_2).\\
    \end{split}
  \end{equation*}
\end{lemma}

\begin{lemma}\label{lem:alpha_bound}
  For Hermitian $H, K$ and $-\pi< H+K < \pi$,
  \begin{equation}
    \label{eq:alpha_bound}
    \begin{split}
      \alpha_{\max} (e^{iH} e^{iK}) & \le \alpha_{\max} (e^{i(H+K)}),\\
      \alpha_{\min} (e^{iH} e^{iK}) & \ge \alpha_{\min} (e^{i(H+K)}).\\
    \end{split}
  \end{equation}
\end{lemma}

\begin{lemma}\label{lem:alpha_tri}
  $\alpha(U_1,U_2) \le \alpha(U_1)+\alpha(U_2)$.
\end{lemma}

\begin{proof}
  If either $\alpha(U_1)$ or $\alpha(U_2)$ equals $\pi$, the above
  equation obviously holds. Now if both $\alpha(U_1)$ and
  $\alpha(U_2)$ is less than $\pi$, we can choose phases $\varphi_1$
  and $\varphi_2$ such that
  \begin{equation*}
    U_1^\dagger = e^{i\varphi_1} V_1, U_2 = e^{i\varphi_2} V_2,
  \end{equation*}
  and $V_1$ and $V_2$ have eigenvalues of arguments in
  $(-\pi/2,\pi/2)$. The condition in Lemma~\ref{lem:cpr} holds for
  $V_1$ and $V_2$ and it follows that $\alpha(V_1V_2) \le
  \alpha(V_1)+\alpha(V_2)$ which finishes the proof by noticing that
  $\alpha(U)$ is invariant under the change of a global phase in $U$.
\end{proof}

It's interesting to note that Lemma~\ref{lem:alpha_tri} implies that
$\alpha(U_1,U_2)$ is a distance measure on the space of $U(d)/U(1)$. Specifically,
\begin{equation*}
  \alpha(U_1,U_3) = \alpha(U_1^\dagger U_2 U_2^\dagger U_3)
  \le \alpha(U_1^\dagger U_2) + \alpha(U_2^\dagger U_3)
  = \alpha(U_1,U_2) + \alpha(U_2,U_3).
\end{equation*}

\begin{lemma}
  \label{lem:alpha_diff}
  For unitary $U$ and $V$, $|\alpha(U)-\alpha(V)| \le \pi \|U-V\|$.
\end{lemma}

\begin{proof}
  Since $\alpha$ is a distance measure,
  \begin{equation*}
    |\alpha(U)-\alpha(V)| = |\alpha(U,I)-\alpha(V,I)| \le \alpha
    (U^\dagger V).
  \end{equation*}
  Using Eq.~\eqref{eq:sin_alpha} and $\sin(x)\ge 2x/\pi$
  for $x\in [0,\pi/2]$, we have
  \begin{equation*}
    \|U-V\| \ge \min_\varphi \|U-e^{i\varphi}V\| \ge
    2\sin\frac{\alpha(U^\dagger V)}{4} \ge \frac{1}{\pi} \alpha(U^\dagger
    V) \ge \frac{1}{\pi} |\alpha(U)-\alpha(V)|.
  \end{equation*}
\end{proof}

\begin{lemma}\label{lem:alpha_t}
  For Hermitian $H, K$ and $0\le H, K, H+K \le \pi$, $0<t<1$,
  \begin{equation}
    \label{eq:alpha_t}
    |\alpha(e^{iHt} e^{iKt}) - \alpha(e^{i(H+K)t})| \le ct^2,
  \end{equation}
  where $c$ is a constant independent of $H, K$ and $t$.
\end{lemma}

\begin{proof}
  Using the expansion of the matrix exponential function and the
  condition $0 \le H,K,H+K \le \pi$, $0<t<1$, it's easy to show that
  there exists some constant $c_1$ such that
  \begin{equation*}
    \| e^{iHt} e^{iKt} - e^{i(H+K)t} \| \le c_1 t^2.
  \end{equation*}
  The inequality follows immediately from
  Lemma~\ref{lem:alpha_diff}.
\end{proof}

With these results in hand, we start the proof of the main result,
Theorem ~\ref{thm:const-depth}.

\begin{proof}[Proof of Theorem~\ref{thm:const-depth}]

  As Non-Identity Check of short circuit is a special case of the
  general Non-Identity Check problem, the fact that it is in
  \class{QMA} follows from the previous result in Ref.~\cite{JWB05}.
  It suffices to prove the hardness result only. We will reduce the
  $1$-D Local Hamiltonian problem to it.

  Suppose we are given an instance of the $1$-D Local Hamiltonian
  problem which has input $H=\sum_{i=1}^r H_i$ and real numbers $a, b$
  with at least inverse polynomial gap. $H$ is a Hamiltonian of an $n$
  particle system with local dimension $d$ and each term $H_i$ is an
  operator on two neighbouring particles which can be described by a
  $d^2$ by $d^2$ Hermitian matrix. It is a ``Yes'' instance if there
  exists some density matrix $\rho$ such that $\tr(H\rho)\le a$ and a
  ``No'' instance if $\tr(H\rho)\ge b$ for all $\rho$. This problem is
  known to be \class{QMA}-Complete for $d\ge12$. For simplicity, one
  can always rescale the problem and assume that $H_i$'s are positive
  semidefinite and $\|H_i\|\le 1$.

  Note that $1$-D property of the problem allows us to write $H$ as
  $H_{\text{odd}} + H_{\text{even}}$ where $H_{\text{odd}}$ and
  $H_{\text{even}}$ each contain local terms acting on different particles.
  This is illustrated in Figure~\ref{fig:1d}. $H_{\text{odd}}$ is the
  sum of $H_1, H_3, H_5, \ldots$ where $H_1$ acts on particle $1$ and
  $2$, $H_3$ acts on particle $3$ and $4$, etc. $H_{\text{even}}$
  consists of $H_2, H_4, \ldots$ where $H_2$ acts on particle $2$ and
  $3$, $H_4$ acts on particle $4$ and $5$, etc.

  \begin{figure}[!hbt]
    \centering
    \includegraphics{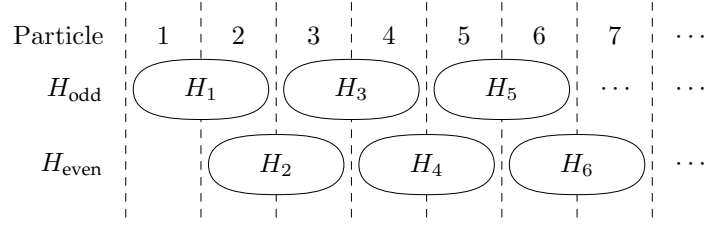}
    \caption{1-D Local Hamiltonian}
    \label{fig:1d}
  \end{figure}

  The first step in the reduction is to modify the Hamiltonian such
  that it will have the maximal possible eigenvalue $r$, where $r$ is
  the number of local terms. To do this, we add an additional
  dimension to each particle and label it with $\ket{d}$ and consider
  the Hamiltonian with local terms
  \begin{equation*}
    \tilde{H}_i = \oprod{d}{d}\otimes\oprod{d}{d} + H_i.
  \end{equation*}
  It should be understood that $H_i$ acts trivially when either
  particle on which it acts is in state $\ket{d}$. Let $\tilde{H}$ be
  the sum $\sum_i\tilde{H}_i$. It's obvious that $\ket{d}^{\otimes n}$
  is an eigenstate of $\tilde{H}$ with eigenvalue $r$ while the
  smallest eigenvalue of $\tilde{H}$ equals that of $H$. The $1$-D
  Hamiltonian problem of $H$ is now reduced to deciding if
  $\lambda(\tilde{H})$, the eigenvalue range of $\tilde{H}$, is larger
  than $r-a$ or smaller than $r-b$. The eigenvalue range problem can
  be viewed as the Hamiltonian version of the Non-Identity Check
  problem for circuits. We will show that it's possible to use circuit
  Non-Identity Check to solve this eigenvalue range problem of local
  Hamiltonian.

  Before further reducing the problem, we normalize $\tilde{H}$ by
  dividing $2r/\pi$ so that the conditions of the lemmas we will use
  are met. Denote the normalized Hamiltonian again with $H$ and its
  local terms with $H_i$ for simplicity; but they are no longer the
  same as they were in the original $1$-D Local Hamiltonian problem.
  After that, we have $\|H\| \le \pi/2$. Let $l$ and $s$ be
  $(r-a)\pi/2r$ and $(r-b)\pi/2r$ respectively. It's a ``Yes''
  instance if $\lambda(H)\ge l$ and a ``No'' instance if
  $\lambda(H)\le s$. Notice that $l$ and $s$ have inverse polynomial
  gap.

  We can now construct a Non-Identity problem as follows. The circuit
  is simply
  \begin{equation}
    \label{eq:U_H}
    U_H = e^{iH_{\text{even}}t} e^{iH_{\text{odd}}t},
  \end{equation}
  and the two threshold real numbers $a=st$, $b=lt-ct^2$. Here, $t$ is
  chosen to be $(l-s)/2c$ where $c$ is the constant in
  Lemma~\ref{lem:alpha_t}. It's easy to check that $b-a$ has at least
  inverse polynomial gap. As the local terms $H_1,H_3,H_5,\ldots$ in
  $H_{\text{odd}}$ are on different particles, $e^{iH_{\text{odd}}t}$
  equals the tensor product of $e^{iH_1t},e^{iH_3t},e^{iH_5t},\ldots$
  and can be implemented in parallel. Similar property holds for
  $H_{\text{even}}$. Therefore $U_H$ is indeed a constant depth
  circuit.

  Since $\lambda(H)$ is promised either larger than $l$ or smaller
  than $s$, we can verify that the promise for the above Non-Identity
  problem also holds. If $\lambda(H)$ is larger than $l$, it follows
  from Lemma~\ref{lem:alpha_t} that $\alpha(U_H)$ is at least
  \begin{equation*}
    \alpha(e^{iHt}) - ct^2 = \lambda(Ht) -ct^2 \ge lt-ct^2 = b.
  \end{equation*}
  If $\lambda(H)$ is smaller than $s$, Lemma~\ref{lem:alpha_bound}
  implies that $\alpha(U_H)$ is at most
  \begin{equation*}
    \alpha(e^{iHt}) \le st = a.
  \end{equation*}
  It's also easy to check that the eigenvalue range problem of $H$ is
  a ``Yes'' (or ``No'') instance if and only if the Non-Identity Check
  problem of $U_H$ is a ``Yes'' (or ``No'') instance.
\end{proof}

It's worth noting that the main idea in the proof is highly related to
quantum simulation using Trotter expansion
\begin{equation*}
  e^{A+B} = \lim_{n\rightarrow\infty} (e^{A/n} e^{B/n})^n.
\end{equation*}
Fortunately however, it is enough to simulate the first round of
$e^{A/n} e^{B/n}$ and leave the amplification procedure to the
verifier.

The circuit we constructed above contains quantum gates such as
$e^{iH_1t}$ which need $\Omega(\log n)$ bits to specify. In order to
translate the result to the case where only a finite universal set of
quantum gates are allowed, we need to expand each gate in the circuit
using Solovay-Kitaev theorem. This will give us the result in
Corollary~\ref{cor:polylog-depth}. The main problem here is to analyze
how the imperfections in each gate will affect the phase range
$\alpha$ of the circuit. Suppose we want to use unitary gates $U_1$
and $U_2$ but the actual implementations are unitary gates $V_1$ and
$V_2$ with $V_1=U_1+E_1$ and $V_2=U_2+E_2$, then
\begin{equation*}
  \begin{split}
    \| V_1V_2-U_1U_2 \| & =
    \| (U_1+E_1)(U_2+E_2)-U_1U_2\|\\
    & = \|U_1E_2+E_1(U_2+E_2)\|\\
    & = \|U_1E_2+E_1V_2\|\\
    & \le \|E_1\|+\|E_2\|,
  \end{split}
\end{equation*}
and similarly,
\begin{equation*}
  \| V_1\otimes V_2-U_1\otimes U_2 \| \le \|E_1\|+\|E_2\|.
\end{equation*}
These two facts and Lemma~\ref{lem:alpha_diff} imply that for any
circuit $C$ and its imperfect implementation $C'$
\begin{equation*}
  |\alpha(C) - \alpha(C')| \le\pi \|C-C'\| \le \pi\sum_i\|E_i\|,
\end{equation*}
where $E_i$'s are the errors in all the gates of $C'$. Thus, the total
error in $\alpha$ of the circuit is at most $\pi$ times the summation
of norms of all errors in each gate. It can be made inverse polynomial
small and much smaller than the gap of threshold parameter $a$ and
$b$. This validates the claim in Corollary~\ref{cor:polylog-depth}.

It's proved in Ref.~\cite{HRC02} that there exists some universal gate
set such that only $O(\log(1/\epsilon))$ number of gates are required
to achieve an error bound of $\epsilon$. A similar argument as above
gives the proof of Corollary~\ref{cor:log-depth}.

\section{Conclusion}

\label{sec:con}

In this paper, we conclude that Non-Identity Check for constant depth
quantum circuit is \class{QMA}-Complete given $\Omega(\log n)$ bit of
precision to each gate. However, the depth may vary when using a fixed
universal gate set. Employing different versions of Solovay-Kitaev
theorem, we are able to prove the hardness for circuits of
polylogarithmic or even logarithmic depth.

It is interesting to compare our result with the problem of
distinguishing mixed state quantum computation in terms of the diamond
norm~\cite{AKN98}. Although the main difference is simply whether some
output are discarded or not, the problem of distinguishing mixed state
quantum computation seems to be much harder. In fact, it was shown to
be \class{QIP}-Complete~\cite{RW05}. Rosgen~\cite{Ros08} further
proved that logarithmic depth quantum circuits are as hard to
distinguish as polynomial depth quantum circuit and thus
distinguishing logarithmic depth mixed state quantum circuits remains
\class{QIP}-Complete.

We leave the question of the complexity of Non-Identity Check for
constant depth quantum circuits with gates from a finite universal
gate set as an interesting open problem.

\section*{Acknowledgment}

The authors would like to thank Daniel Gottesman, Richard Cleve and
John Watrous for helpful discussions. Research at Perimeter Institute
is supported by the Government of Canada through Industry Canada and
by the Province of Ontario through the Ministry of Research \&
Innovation.

\small
\bibliographystyle{abbrv}
\bibliography{qma-short}

%% Start-Of-Trailer
\end{document}